\documentclass[conference]{IEEEtran}
\usepackage{cite}
\usepackage{amsmath,amssymb,amsfonts}
\usepackage{amsthm}
\usepackage{algorithmic}
\usepackage{textcomp}
\usepackage{xcolor}
\usepackage[pdftex]{graphicx}
\usepackage[latin1]{inputenc}
\usepackage{array}
\usepackage{booktabs}
\usepackage{subfig}
\usepackage{bm}
\usepackage{adjustbox}
\usepackage{hyperref}

\def\BibTeX{{\rm B\kern-.05em{\sc i\kern-.025em b}\kern-.08em
    T\kern-.1667em\lower.7ex\hbox{E}\kern-.125emX}}

\renewcommand{\vec}[1]{\bm{#1}}
\newcommand{\mat}[1]{\bm{#1}}

\newcommand{\trn}[1]{#1^\intercal}

\newcommand{\ipt}[2]{\trn{#1} #2}

\newcommand{\dsq}[2]{\bigl \lVert #1 - #2 \bigr \rVert^2}

\newcommand{\st}{\operatorname{s.\!t.}}

\newcommand{\amin}[1]{\operatorname*{argmin}_{#1}}

\newtheorem{thm}{Theorem}
\newtheorem{lem}[thm]{Lemma}

\definecolor{ml2rblu}{rgb}{0.02,0.27,0.45}
\definecolor{ml2ryel}{rgb}{0.98,0.72,0.18}
\definecolor{ml2rgrn}{rgb}{0.50,0.71,0.18}
\definecolor{ml2rtrq}{rgb}{0.00,0.57,0.57}

\hypersetup{%
  pdftitle={QUBOs for Sorting Lists and Building Trees},
  pdfauthor={Christian Bauckhage, Thore Gerlach, Nico Piatkowski},
  pdfsubject={quadratic unconstrained binary optimization},
  pdfkeywords={abstract data structures, neurocomputing, quantum computing},
  colorlinks=true,
  urlcolor=ml2rtrq,
  linkcolor=ml2rtrq,
  citecolor=ml2rtrq,
  bookmarks=false}

\begin{document}

\title{QUBOs for Sorting Lists and Building Trees}

\author{\IEEEauthorblockN{Christian Bauckhage}
\IEEEauthorblockA{\textit{Computer Science}\\
\textit{University of Bonn}\\
Bonn, Germany}
\and
\IEEEauthorblockN{Thore Gerlach}
\IEEEauthorblockA{\textit{Media Engineering}\\
\textit{Fraunhofer IAIS}\\
St.~Augustin, Germany}
\and
\IEEEauthorblockN{Nico Piatkowski}
\IEEEauthorblockA{\textit{Media Engineering}\\
\textit{Fraunhofer IAIS}\\
St.~Augustin, Germany}}

\maketitle

\begin{abstract}
We show that the fundamental tasks of sorting lists and building search trees or heaps can be modeled as quadratic unconstrained binary optimization problems (QUBOs). The idea is to understand these tasks as permutation problems and to devise QUBOs whose solutions represent appropriate permutation matrices. We discuss how to construct such QUBOs and how to solve them using Hopfield nets or (adiabatic) quantum computing. In short, we show that neurocomputing methods or quantum computers can solve problems usually associated with abstract data structures. 
\end{abstract}


\section{Introduction}

In this paper, we are concerned with quadratic unconstrained binary optimization problems (QUBOs) of the form
\begin{equation}
\label{eq:QUBO}
\vec{z}_* = \amin{\vec{z} \in \{ 0, 1 \}^N} \, \trn{\vec{z}} \mat{R} \, \vec{z} + \ipt{\vec{r}}{\vec{z}}
\end{equation}
where the objective is to find an optimal vector $\vec{z}_*$ of $N$ binary decision variables and where $\mat{R} \in \mathbb{R}^{N \times N}$ and $\vec{r} \in \mathbb{R}^N$ contain application specific parameters. 

QUBOs are surprisingly versatile and occur in numerous settings \cite{Kochenberger2014-TUB,Lucas2014-IFO,Calude2017-QUBO,Glover2018-ATO,Bauckhage2019-AQF,Muecke2019-LBB,Bauckhage2020-AQC,Date2021-QFF,Biesner2022-SSS}. Notable use cases include RNA folding, budget allocation, portfolio optimization, routing, location planning, or item diversification. More generally, QUBOs arise in clique finding, graph partitioning, satisfiablity testing, data clustering, or classifier training. 

In either case, QUBOs constitute combinatorial optimization problems. Indeed, as they deal with binary decision variables, QUBOs are specific integer programming problems and thus NP-hard in general. Yet they are also isomorphic to Hopfield- or Ising energy minimization problems known from neuro- or adiabtic quantum computing \cite{Hopfield1982-NNA,Farhi2000-QCB,Johnson2011-QAW,Farhi2000-QAOA,Bauckhage2020-PSW}. Since solvers can thus be implemented on emerging, potentially superior platforms such as neuromorphic- or quantum computers, research on the general capabilities and merits of QUBOs is increasing. 

The work reported here falls into this category. We are interested in the universality of QUBOs and explore their use in tasks rarely seen as combinatorial optimization problems. In particular, we consider the basic computer scientific problems of sorting and tree building. Since there exist well established classical algorithms for this purpose, our study is a study of principle. We demonstrate that problems which usually involve the manipulation of abstract data structures can also be expressed as QUBOs of the form in \eqref{eq:QUBO}. This, in turn, establishes that neuromorphic- or quantum computers can sort and build trees.

\subsection{Basic Ideas and Notation}

Throughout, we assume we are given an unordered list or sequence $x_1, x_2, \ldots, x_n$ of $n$ numbers $x_i \in \mathbb{R}$ which we gather in a vector $\vec{x} \in \mathbb{R}^n$. The basic idea is then to determine an $n \times n$ permutation matrix $\mat{P}$ such that the entries of the vector $\vec{y} = \mat{P} \vec{x} \in \mathbb{R}^n$ are arranged in a certain, problem specific order. We therefore recall that a permutation matrix is a square binary matrix whose rows and columns sum to one. Formally, we express these characteristics as
\begin{align}
\mat{P} & \in \{0, 1\}^{n \times n} \\
\mat{P} \,    \vec{1} & = \vec{1} \\
\trn{\mat{P}} \vec{1} & = \vec{1}
\end{align}
where $\vec{1}  \in \mathbb{R}^n$ denotes the vector of all $1$s.

Our main contribution is to show that searching for an appropriate permutation matrix is tantamount to solving a QUBO as in \eqref{eq:QUBO} where $N = n^2$. Our derivation will involve the operation $\vec{v} = \operatorname*{vec}(\mat{M})$ which vectorizes a matrix $\mat{M}$ by stacking its columns into a vector $\vec{v}$. The operation $\mat{M} = \operatorname*{mat}(\vec{v})$ reverses this operation and matricizes $\vec{v}$ into $\mat{M}$. Finally, our derivation will also involve Kronecker products of the $n \times n$ identity matrix $\mat{I}$ and $n$-dimensional vectors; these Kronecker products are denoted by $\otimes$.

\subsection{Overview}

Next, in Section~\ref{sec:sorting}, we devise a QUBO formulation of the sorting problem. This will form the basis for our discussion in Section~\ref{sec:searching} where we set up QUBOs for constructing abstract data structures such as trees and heaps. In Section~\ref{sec:examples}, we present baseline experiments in which we use Hopfield networks to solve sorting and tree building QUBOs. Our results corroborate the viability of the proposed modeling framework. Finally, in Section~\ref{sec:conclusion}, we summarize key findings and discuss potential implications.

\section{QUBOs for Sorting Lists}
\label{sec:sorting}

In this section, we show that the basic computer scientific problem of sorting an unordered list of numbers can be cast as a QUBO. Without loss of generality, we focus on sorting in ascending order. Our discussion will be rather detailed because, once we have established that QUBOs can be used for sorting in ascending order, it will be easy to see how to adapt them to sorting in orders which represent serialized trees or heaps.

To begin with, we gather the given unordered numbers into an $n$-dimensional vector
\begin{equation}
\vec{x} = \trn{[x_1, x_2, \ldots, x_n]}
\end{equation}
This innocuous preparatory step allows us to treat the sorting problem as the problem of finding an $n \times n$ permutation matrix $\mat{P}$ such that the entries of the permuted vector $\vec{y} = \mat{P} \vec{x}$ 
obey $y_1 \leq y_2 \leq \ldots \leq y_n$. 

In order to devise an objective function whose minimization would yield the sought after permutation matrix, we introduce yet another, rather specific $n$-dimensional vector, namely
\begin{equation}
\vec{n} = \trn{[1, 2, \ldots, n]}
\end{equation}

As the entries of this auxiliary vector are characterized by $n_1 \leq n_2 \leq \ldots \leq n_n$, we next appeal to the \emph{rearrangement inequality} of Hardy, Littlewood, and Polya \cite{Hardy1952-I}. This classical result implies that the negated inner product $-\ipt{\vec{y}}{\vec{n}}$ is minimal whenever $y_1 \leq y_2 \leq \ldots \leq y_n$. In other words, the expression $-\ipt{\vec{y}}{\vec{n}}$ is minimal if the entries of $\vec{y}$ and $\vec{n}$ are similarly sorted. Moreover, since we defined $\vec{y} = \mat{P} \vec{x}$, it follows that a permutation matrix $\mat{P}$ which sorts 
the entries of $\vec{x}$ can be found by solving 
\begin{equation}
\label{eq:qubo-step1}
\begin{aligned} 
\mat{P} = \amin{\mat{Z} \in \{0,1\}^{n \times n}} \, & \, -\trn{\vec{x}} \trn{\mat{Z}} \vec{n} \\
       \st \quad & \;\;
       \begin{aligned}
       \mat{Z} \,    \vec{1} & = \vec{1} \\
       \trn{\mat{Z}} \vec{1} & = \vec{1}
      \end{aligned}
\end{aligned}
\end{equation} 

This intermediate (arguably lesser known) result shows that sorting can indeed be understood as an optimization problem \cite{Bauckhage2021-SAL,Aspnes2004-NOL}. However, \eqref{eq:qubo-step1} is a linear programming problem over binary matrices rather than a QUBO over binary vectors. Next, we therefore transform it into a problem of the form in \eqref{eq:QUBO}. This will involve two major steps: First, we rewrite the linear program over binary matrices as a linear program over binary vectors and, second, express the latter as a QUBO.

To rewrite our matrix linear program as a vector linear program, we resort to a lemma which we state and prove in the \hyperref[sec:appendix]{Appendix}. With respect to the matrix vector product $\trn{\mat{Z}} \vec{n}$ which occurs in the objective of \eqref{eq:qubo-step1}, the lemma establishes that vectorizing $\trn{\mat{Z}} \in \{0,1\}^{n \times n}$ into $\vec{z} \in \{0, 1\}^{n^2}$ and matricizing $\vec{n} \in \mathbb{R}^n$ into $\mat{N} \in \mathbb{R}^{n \times n^2}$ by means of
\begin{align}
\vec{z} & = \operatorname*{vec}\,(\mat{Z}) \\
\mat{N} & = \mat{I} \otimes \trn{\vec{n}} \label{eq:matN}
\intertext{provides us with the following identity}
\trn{\mat{Z}} \vec{n} & = \mat{N} \, \vec{z} \label{eq:qubo-step2}
\end{align}
In other words, we can equivalently express the product of the $n \times n$ matrix $\trn{\mat{Z}}$ and the $n$-dimensional vector $\vec{n}$ as a product of an $n \times n^2$ matrix $\mat{N}$ and an $n^2$-dimensional vector $\vec{z}$. 

Similar arguments apply to the expressions in the equality constraints in \eqref{eq:qubo-step1}. That is, we further have
\begin{align}
     \mat{Z} \, \vec{1} & = \mat{C}_r \, \vec{z} \\
\trn{\mat{Z}}   \vec{1} & = \mat{C}_c \, \vec{z}
\intertext{where the $n \times n^2$ matrices $\mat{C}_r$ and $\mat{C}_c$ on the right hand sides capture row- and column sum constraints and are given by}
\mat{C}_r & = \trn{\vec{1}} \otimes \mat{I} \label{eq:matCr} \\
\mat{C}_c & = \mat{I} \otimes \trn{\vec{1}} \label{eq:matCc}
\end{align}

Consequently, we can reformulate the problem of estimating an optimal permutation matrix as the problem of finding an optimal binary vector, namely
\begin{equation}
\label{eq:qubo-step3}
\begin{aligned} 
\vec{z}_* = \amin{\vec{z} \in \{0,1\}^{n^2}} \, & \, -\trn{\vec{x}} \mat{N} \, \vec{z} \\
       \st \quad & \;\;
       \begin{aligned}
       \mat{C}_r \, \vec{z} & = \vec{1} \\
       \mat{C}_c \, \vec{z} & = \vec{1}
      \end{aligned}
\end{aligned}
\end{equation}
 
Given this problem, we note that, once it has been solved for $\vec{z}_*$, the actually sought after permutation matrix can be computed as $\mat{P} = \trn{\operatorname{mat}(\vec{z}_*)}$ to then obtain the sorted version $\vec{y} = \mat{P} \vec{x}$ of $\vec{x}$.

In order to turn the intermediate linear constrained binary optimization problem in \eqref{eq:qubo-step3} into a quadratic unconstrained optimization binary problem, we note the equivalencies
\begin{alignat}{3}
\mat{C}_r \, \vec{z} & = \vec{1} && \;\;\Leftrightarrow\;\; \dsq{\,\mat{C}_r \, \vec{z}}{\vec{1}} & \; = 0 \\
\mat{C}_c \, \vec{z} & = \vec{1} && \;\;\Leftrightarrow\;\; \dsq{\,\mat{C}_c \, \vec{z}}{\vec{1}} & \; = 0
\end{alignat}
and expand the squared Euclidean distances as
\begin{align}
\dsq{\,\mat{C}_r \vec{z}}{\vec{1}} & = \trn{\vec{z}} \trn{\mat{C}_r} \mat{C}_r \vec{z} - 2 \, \trn{\vec{1}} \mat{C}_r \vec{z} + \ipt{\vec{1}}{\vec{1}} \\
\dsq{\,\mat{C}_c \vec{z}}{\vec{1}} & = \trn{\vec{z}} \trn{\mat{C}_c} \mat{C}_c \vec{z} - 2 \, \trn{\vec{1}} \mat{C}_c \vec{z} + \ipt{\vec{1}}{\vec{1}}
\end{align}

Since $\ipt{\vec{1}}{\vec{1}} = n$ is a constant independent of $\vec{z}$, we therefore have the following Lagrangian for the problem in \eqref{eq:qubo-step3}
\begin{align}
L \bigl( \vec{z}, \lambda_r, \lambda_c \bigr) 
= & -\trn{\vec{x}} \mat{N} \, \vec{z} \notag \\
& + \lambda_r \bigl[ \trn{\vec{z}} \trn{\mat{C}_r} \mat{C}_r \vec{z} - 2 \, \trn{\vec{1}} \mat{C}_r \vec{z} \bigr] \notag \\
& + \lambda_c \,\bigl[ \trn{\vec{z}} \trn{\mat{C}_c} \mat{C}_c \vec{z} - 2 \, \trn{\vec{1}} \mat{C}_c \vec{z} \bigr] \\[2ex]
= & \; \trn{\vec{z}} \bigl[ \lambda_r \, \trn{\mat{C}_r} \mat{C}_r + \lambda_c \, \trn{\mat{C}_c} \mat{C}_c \bigr] \vec{z} \notag \\
& \, - \Bigl[ \trn{\vec{x}} \mat{N} + 2 \, \trn{\vec{1}} \bigl[ \lambda_r \, \mat{C}_r + \lambda_c \, \mat{C}_c \bigr] \Bigr] \,  \vec{z} 
\end{align}
Here, $\lambda_r$ and $\lambda_c$ are Lagrange multipliers which we henceforth treat as parameters that have to be set manually. (In section~\ref{sec:examples}, we suggest a simple, problem independent mechanism for this purpose.)

Finally, upon introducing the following $n^2 \times n^2$ matrix and the following $n^2$-dimensional vector
\begin{align}
\mat{R} & = \lambda_r \, \trn{\mat{C}_r} \mat{C}_r + \lambda_c \, \trn{\mat{C}_c} \mat{C}_c \label{eq:matR} \\
\vec{r} & = -\trn{\mat{N}} \vec{x} - 2 \, \trn{\bigl[ \lambda_r \, \mat{C}_r + \lambda_c \, \mat{C}_c \bigr]} \vec{1} \label{eq:vecR}
\end{align}
we find that \eqref{eq:qubo-step3} is equivalent to the following QUBO
\begin{equation}
\label{eq:qubo-step4}
\vec{z}_* = \amin{\vec{z} \in \{ 0, 1 \}^{n^2}} \; \trn{\vec{z}} \mat{R} \, \vec{z} + \ipt{\vec{r}}{\vec{z}} 
\end{equation}

Again, if we could solve \eqref{eq:qubo-step4} for $\vec{z}_*$, the actually sought after permutation matrix is $\mat{P} = \trn{\operatorname{mat}(\vec{z}_*)}$ and allows us to compute the sorted version $\vec{y} = \mat{P} \vec{x}$ of $\vec{x}$.

\section{QUBOs for Building Trees and Heaps}
\label{sec:searching}

\begin{figure}[t]
\centering
\subfloat[binary search tree \label{fig:treexmpl-t}]{%
	\includegraphics[width=0.55\columnwidth]{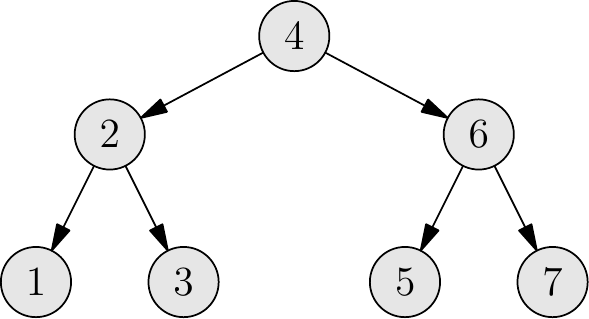}}

\subfloat[serialization of the tree via breadth-first traversal\label{fig:treexmpl-l}]{%
	\includegraphics[width=0.9\columnwidth]{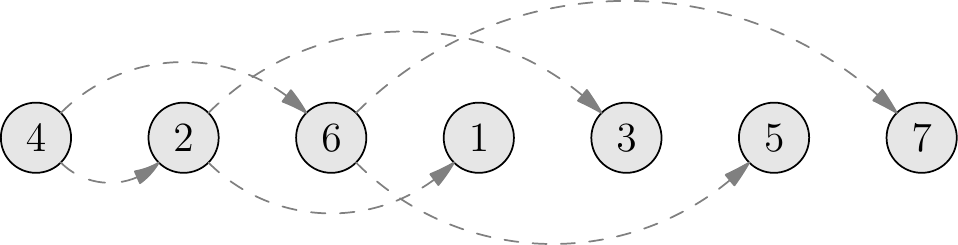}}
\caption{\label{fig:treexmpl} A binary search search tree over the numbers $1, \ldots, 7$.}
\end{figure}

The fact that QUBOs can be used to determine permutation matrices $\mat{P}$ which arrange the entries of a given vector $\vec{x}$ in an order prescribed by another vector $\vec{n}$ has applications beyond conventional sorting. It does, for instance, also allow for arranging the entries of $\vec{x}$ into more abstract data structures such as binary search trees. The general idea behind this claim is best explained by means of an example.

Figure~\ref{fig:treexmpl-t} shows the numbers $1$ through $7$ stored in a binary search tree. Looking at this figure, we recall that such a tree is a labeled directed acyclic graph whose vertices have up to two successors. Moreover, vertex labels of a binary search tree are arranged in such a manner that the label of each internal vertex is greater than the label of its left successor and less than the label of its right successor.

Figure~\ref{fig:treexmpl-l} shows a serialization of the tree in Fig.~\ref{fig:treexmpl-t}. This serialization resulted from a breadth-first traversal of the tree and is structure preserving in the following sense: Letting $v_0, v_1, v_2, \ldots$ denote the serialized vertices (where we deliberately start counting at $0$), the tree can be recovered by choosing the left and right successors of vertex $v_i$ to be $v_j$ and $v_k$ with $j = b i + 1$ and $k = b i + 2$ where $b=2$ denotes the branching factor of the tree. In Fig.~\ref{fig:treexmpl-l}, these implicit successor relations are visualized by means of dashed arrows.

Overall, our example illustrates that binary search trees can be thought of as specifically ordered lists. Hence, sticking with our example, if we wanted to arrange $n=7$ arbitrary numbers $x_1, \ldots, x_7$ into a search tree, we could gather them in a vector $\vec{x} \in \mathbb{R}^7$, consider the auxiliary vector 
\begin{equation}
\vec{n} = \trn{[4, 2, 6, 1, 3, 5, 7]}
\end{equation}
and set up a QUBO as in \eqref{eq:qubo-step4} to determine a permutation matrix that arranges the $x_i$ into the desired tree order.

Of course this approach extends to settings with arbitrary many numbers as well as to trees with branching factors greater than two. It also extends to different labeling schemes and thus to different abstract data structures. We once again illustrate this claim by means of an example.

\begin{figure}[t]
\centering
\subfloat[maximum heap \label{fig:heapxmpl-t}]{%
	\includegraphics[width=0.55\columnwidth]{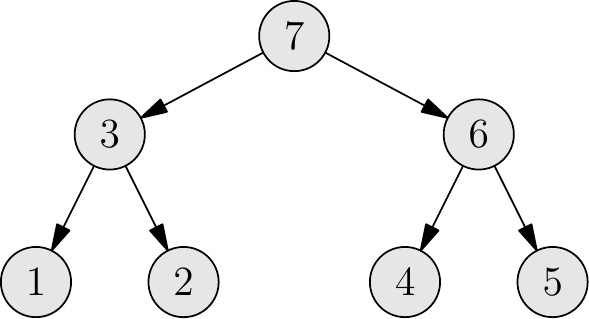}}

\subfloat[serialization of the heap via breadth-first traversal \label{fig:heapxmpl-l}]{%
	\includegraphics[width=0.9\columnwidth]{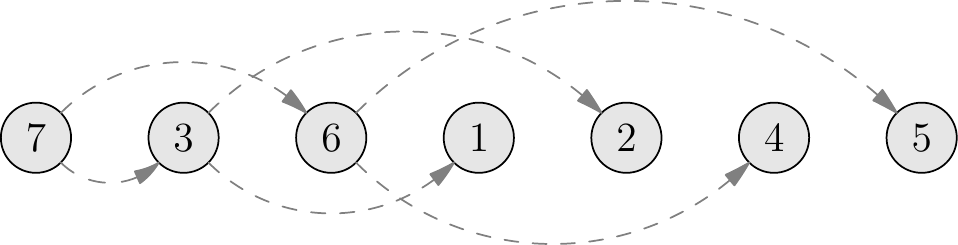}}
\caption{\label{fig:heapxmpl} A maximum heap over the numbers $1, \ldots, 7$.}
\end{figure}

Figure~\ref{fig:heapxmpl-t} shows the numbers $1$ through $7$ arranged in form of a maximum heap. Here, we recall that a maximum heap is a labeled binary tree where the labels of internal vertices are greater than or equal to the labels of their successors. Figure~\ref{fig:heapxmpl-l} shows a serialization of the heap in Fig.~\ref{fig:heapxmpl-t} again obtained from  breadth-first traversal. Since this serialization is once again structure preserving, the above arguments reapply.

Hence, if we wanted to arrange $n=7$ arbitrary numbers $x_1, \ldots, x_7$ into a maximum heap, we could gather them in a vector $\vec{x} \in \mathbb{R}^7$, consider the auxiliary vector 
\begin{equation}
\vec{n} = \trn{[7, 3, 6, 1, 2, 4, 5]}
\end{equation}
and set up a QUBO as in \eqref{eq:qubo-step4} to determine a permutation matrix that arranges the $x_i$ into the desired heap order.

The two tree building examples in this section thus reveal the sorting QUBO in \eqref{eq:qubo-step4} to be a general purpose tool for ordering or structuring problems. In a sense, it can actually be understood as a \emph{programmable machine}. The input variables it processes are contained in vector $\vec{x}$ and the program it executes is given by vector $\vec{n}$. Different programs, i.e.~different choices of $\vec{n}$, will cause the \emph{programmable} QUBO to produce different results of different utility. 

\section{Practical Examples}
\label{sec:examples}

In this section, we present simple experiments which verify that the above theory can be put into practice.

We first address the open question of how to choose the two free parameters in equations \eqref{eq:matR} and \eqref{eq:vecR}. We then recap our construction of sorting or reordering QUBOs over binary variables and recall how to turn those into QUBOs over bipolar variables. These can be solved using (adiabatic) quantum computers or Hopfield nets and we apply the latter for sorting, tree building, and heap building. 

\begin{figure*}[t!]
\centering
\subfloat[state evolution of a Hopfield net that solves a sorting QUBO]{%
\scriptsize
\begin{tabular}{rcr}
$t$ & $\vec{s}_t$ & $E \bigl( \vec{s}_t \bigr)$ \\
\midrule
$   0$ & $- - - - - - - - - - - - - - - - - - - - - - - - - - - - - - - - - - - - - - - - - - - - - - - - -$  & $-673.5$ \\
$   1$ & $- - - - - - - - - - - - - + - - - - - - - - - - - - - - - - - - - - - - - - - - - - - - - - - - -$  & $-689.1$ \\
$   2$ & $- - - - - - - - - - - - - + - - - - - - - - - - - - - - - - - - - - - - - - - - + - - - - - - - -$  & $-704.4$ \\
$   3$ & $- - - - + - - - - - - - - + - - - - - - - - - - - - - - - - - - - - - - - - - - + - - - - - - - -$  & $-719.4$ \\
$   4$ & $- - - - + - - - - - - - - + - - - - - - - - - - + - - - - - - - - - - - - - - - + - - - - - - - -$  & $-734.0$ \\
$   5$ & $- - - - + - - - - - - - - + - - - - - - - - - - + - - - - - - - - - - - - - - - + - - - + - - - -$  & $-748.3$ \\
$   6$ & $- - - - + - - - - - - - - + - - - - - - - - - - + - - - - + - - - - - - - - - - + - - - + - - - -$  & $-762.4$ \\
$   7$ & $- - - - + - - - - - - - - + + - - - - - - - - - + - - - - + - - - - - - - - - - + - - - + - - - -$  & $-776.4$ \\
$   8$ & $- - - - + - - - - - - - - + + - - - - - - - - - + - - - - + - - - - - - - - - - + - - - + - - - -$  & $-776.4$ \\
\end{tabular}}

\subfloat[state evolution of a Hopfield net that solves a tree building QUBO]{%
\scriptsize
\begin{tabular}{rcr}
$t$ & $\vec{s}_t$ & $E \bigl( \vec{s}_t \bigr)$ \\
\midrule
$   0$ & $- - - - - - - - - - - - - - - - - - - - - - - - - - - - - - - - - - - - - - - - - - - - - - - - -$  & $-673.5$ \\
$   1$ & $- - - - - - - - - - - - - + - - - - - - - - - - - - - - - - - - - - - - - - - - - - - - - - - - -$  & $-689.1$ \\
$   2$ & $- - - - - - - - - - - - - + - - - - - - - - - - - - - - - - - - - - - - - + - - - - - - - - - - -$  & $-704.4$ \\
$   3$ & $- - - - - + - - - - - - - + - - - - - - - - - - - - - - - - - - - - - - - + - - - - - - - - - - -$  & $-719.4$ \\
$   4$ & $- - - - - + - - - - - - - + - - - - - - - + - - - - - - - - - - - - - - - + - - - - - - - - - - -$  & $-734.0$ \\
$   5$ & $- - - - - + - - - - - - - + - - - - - - - + - - - - - - - - - - - - - - - + - - - - - - - - + - -$  & $-748.3$ \\
$   6$ & $- - - - - + - - - - - - - + - - - - - - - + - - - - - - - + - - - - - - - + - - - - - - - - + - -$  & $-762.4$ \\
$   7$ & $- - - - - + - - - - - - - + - - - + - - - + - - - - - - - + - - - - - - - + - - - - - - - - + - -$  & $-776.4$ \\
$   8$ & $- - - - - + - - - - - - - + - - - + - - - + - - - - - - - + - - - - - - - + - - - - - - - - + - -$  & $-776.4$ \\
\end{tabular}}

\subfloat[state evolution of a Hopfield net that solves a heap building QUBO]{%
\scriptsize
\begin{tabular}{rcr}
$t$ & $\vec{s}_t$ & $E \bigl( \vec{s}_t \bigr)$ \\
\midrule
$   0$ & $- - - - - - - - - - - - - - - - - - - - - - - - - - - - - - - - - - - - - - - - - - - - - - - - -$  & $-673.5$ \\
$   1$ & $- - - - - - - + - - - - - - - - - - - - - - - - - - - - - - - - - - - - - - - - - - - - - - - - -$  & $-689.1$ \\
$   2$ & $- - - - - - - + - - - - - - - - - - - - - - - - - - - - - - - - - - - - - + - - - - - - - - - - -$  & $-704.4$ \\
$   3$ & $- - - - - - + + - - - - - - - - - - - - - - - - - - - - - - - - - - - - - + - - - - - - - - - - -$  & $-719.4$ \\
$   4$ & $- - - - - - + + - - - - - - - - - - - - - - - - - - + - - - - - - - - - - + - - - - - - - - - - -$  & $-734.0$ \\
$   5$ & $- - - - - - + + - - - - - - - - - - - - - - - - - - + - - - - - - - - - - + - - - - - + - - - - -$  & $-748.3$ \\
$   6$ & $- - - - - - + + - - - - - - - - - - - - - - - - - - + - - - - - + - - - - + - - - - - + - - - - -$  & $-762.4$ \\
$   7$ & $- - - - - - + + - - - - - - - - - + - - - - - - - - + - - - - - + - - - - + - - - - - + - - - - -$  & $-776.4$ \\
$   8$ & $- - - - - - + + - - - - - - - - - + - - - - - - - - + - - - - - + - - - - + - - - - - + - - - - -$  & $-776.4$ \\
\end{tabular}}
\caption{\label{fig:runs} Visualizations of the evolution of states and energies of Hopfield nets of $N = n^2$ neurons which solve sorting or reordering QUBOs for $\vec{x} \in \mathbb{R}^n$. In each example, $n=7$ and all $N=49$ neurons are initially inactive. Over time $t$, the state $\vec{s}_t$ of the networks changes, i.e.~neurons switch from inactive ($-$) to active ($+$), and the energy $E(\vec{s}_t)$ decreases. Each process converges within only $O(n)$ steps to a stable state which encodes the solution to the respective problem.}
\end{figure*}

\subsection{On Choosing $\lambda_r$ and $\lambda_c$}

Working with parameterized QUBOs requires experience or guidelines as to suitable parameterizations. With regard to the Lagrange parameters $\lambda_r$ and $\lambda_c$ in \eqref{eq:matR} and \eqref{eq:vecR}, we therefore note that they should be chosen such that the objective function in \eqref{eq:qubo-step4} is balanced. This is to say that the contribution of the term $-\trn{\vec{x}} \mat{N} \, \vec{z}$ should not outweigh the sum-to-one constraints $\mat{C}_r \, \vec{z} = \vec{1}$ and $\mat{C}_c \, \vec{z} = \vec{1}$. Since this may happen whenever the sum (of the norms) of the entries of $\vec{x}$ is larger than $n$, a simple idea is to work with an $L_1$-normalized version $\vec{x}'$ of $\vec{x}$ where
\begin{equation}
\label{eq:normalization}
x_i' = \frac{x_i}{\sum_j \lvert x_j \rvert}
\end{equation}
Using this normalization of the problem input, we may choose
\begin{equation}
\label{eq:parameters}
\lambda_r = \lambda_c = n
\end{equation}
to appropriately weigh the individual components of the overall minimization objective.

\subsection{Setting Up Sorting / Ordering QUBOs}

Given an input vector $\vec{x}$ whose entries are to be rearranged into a certain order, we require an auxiliary vector $\vec{n}$ whose entries reflect that order. We then normalize $\vec{x}$ using \eqref{eq:normalization} and set $\lambda_r$ and $\lambda_c$ according to \eqref{eq:parameters}. Given these preparations, we compute $\mat{N}$, $\mat{C}_r$, and $\mat{C}_c$ using \eqref{eq:matN}, \eqref{eq:matCr}, and \eqref{eq:matCc}. This finally allows for computing $\mat{R}$ and $\vec{r}$ in \eqref{eq:matR} and \eqref{eq:vecR} which parameterize the QUBO in \eqref{eq:qubo-step4}. Note, however, that $\vec{r}$ is now computed with respect to $\vec{x}'$ instead of $\vec{x}$.

\subsection{From Binary to Bipolar QUBOs}

Recall that binary and bipolar vectors are affine isormorphic. That is, if $\vec{z} \in \{0,1\}^N$ is binary, then
$\vec{s} = 2 \, \vec{z} - \vec{1}$ is bipolar. Likewise, if $\vec{s} \in \{-1,+1\}^N$ is bipolar, then $\vec{z} = (\vec{s} + \vec{1}) / 2$ is binary. QUBOs as in \eqref{eq:qubo-step4} can therefore also be expressed as 
\begin{align}
\label{eq:Ising}
\vec{s}_* & = \amin{\vec{s} \in \{ \pm 1 \}^N} \, \trn{\vec{s}} \mat{Q} \, \vec{s} + \ipt{\vec{q}}{\vec{s}}
\end{align}
where $\mat{Q} = \tfrac{1}{4} \, \mat{R}$ and $\vec{q} = \tfrac{1}{2} \, \mat{R} \, \vec{1} + \tfrac{1}{2} \, \vec{r}$. This is of considerable interest because \eqref{eq:Ising} is an Ising energy minimization problem that can be solved on adiabatic quantum computers or --using the quantum approximate optimization algorithm-- on quantum gate computers. 

\subsection{Solving Sorting / Ordering QUBOs with Hopfield Nets}

Substituting $\mat{W} = -2 \, \mat{Q}$ and $\vec{\theta} = \vec{q}$, we may alternatively consider
\begin{align}
\label{eq:Hopfield}
\vec{s}_* & = \amin{\vec{s} \in \{ \pm 1 \}^N} -\tfrac{1}{2} \, \trn{\vec{s}} \mat{W} \vec{s} + \ipt{\vec{\theta}}{\vec{s}}
\end{align}
where $\mat{W}$ and $\vec{\theta}$ can be thought of the weight matrix and bias vector of a Hopfield network. Hence, the minimization objective in \eqref{eq:Hopfield} now defines the energy of a Hopfield network in state $\vec{s}$. 

As Hopfield energies go, the energy landscape of a Hopfield net for sorting or reordering is fairly well behaved. That is, it does not suffer from (spurious) local minima which reflects the fact that sorting or reordering do not constitute hard problems. 

When running a Hopfield net for sorting or reordering, we may thus consider a steepest energy descent mechanism 
to update the network state $\vec{s}_t$ in iteration $t$ \cite{Bauckhage2021-HnetSort}. Our results below indicate that this yields fast convergence to a state which represents the solution to the problem at hand.

\subsection{Exemplary Results}

Just as in the previous section, we focus on simple settings and experiment with $n=7$ numbers. In particular, we consider 
\begin{equation}
\vec{x} = \trn{[46, 52, -12, 33, 10, 51, 24]}
\end{equation} 

In order to sort these numbers in ascending order and to arrange them into a binary tree or maximum heap, we work with the auxiliary vectors
\begin{align}
\vec{n}_S & = \trn{[1, 2, 3, 4, 5, 6, 7]} \\
\vec{n}_T & = \trn{[4, 2, 6, 1, 3, 5, 7]} \\
\vec{n}_H & = \trn{[7, 3, 6, 1, 2, 4, 5]}
\end{align}
and prepare corresponding Hopfield nets. In each experiment, we set the initial network state to $\vec{s}_0 = - \vec{1}$ and then run the steepest energy descent update mechanism.

Figure~\ref{fig:runs} visualizes the evolution of the respective Hopfield nets. It shows that each network rapidly converges to a stable state $\vec{s}_\infty$. In fact, all three networks converge within $O(n)$ updates which once again indicates that neither sorting nor reordering are difficult problems.

Moreover, the stable states the networks converge to encode the optimal solution $\vec{s}_*$ to the respective sorting, tree-, or heap building QUBO. That is, all three stable states encode a permutation matrix $\mat{P} = \trn{\operatorname{mat}\bigl( (\vec{s}_* + \vec{1}) / 2 \bigr)}$ which permutes $\vec{x}$ into the desired order. 
 
To be specific, we obtain permutation matrices $\mat{P}_S$, $\mat{P}_T$, and $\mat{P}_H$ which produce 
\begin{alignat}{2}
\vec{y}_S & = \mat{P}_S \, \vec{x} && = \trn{[-12, 10, 24, 33, 46, 51, 52]} \\
\vec{y}_T & = \mat{P}_T \, \vec{x} && = \trn{[33, 10, 51, -12, 24, 46, 52]} \\
\vec{y}_H & = \mat{P}_H \, \vec{x} && = \trn{[52, 24, 51, -12, 10, 33, 46]}
\end{alignat}
and thus correctly solve the problems we considered.

\section{Conclusion}
\label{sec:conclusion}

This paper resulted from ongoing research in which we ask: How universal is quadratic unconstrained binary optimization? 

Examining the use of QUBOs for problems not commonly seen as combinatorial, we looked at the fundamental tasks of sorting lists and building trees. While there exists numerous established and efficient algorithms for these problems, our interest was in re-conceptualizing them because, if they can be modeled as QUBOs, it is clear that they can be solved using neuromputing techniques or quantum computing.

Our key contribution was to show that sorting and tree building can indeed be cast as QUBOs. 

The basic idea was to simply understand these problems as permutation problems and to devise objective functions whose minimization results in appropriate permutation matrices. We first appealed to the rearrangement inequality and expressed sorting or reordering as linear programming problems over binary matrices. Using linear algebraic arguments and standard tools from optimization theory, we then showed how to rewrite them as linear program over binary vectors and, subsequently, as QUBOs.

Experiments demonstrated that the QUBO we derived in \eqref{eq:qubo-step4} provides a general purpose model for ordering or structuring problems. In a sense, it can be seen as a \emph{programmable machine}. The input this machine processes consists of a vector of numbers to be be reordered; the program it executes is given by another vector which dictates the order to be produced. As such orders may represent trees or heaps or other kinds of data types, QUBOs and, consequently, neuromorphic- or quantum computers can manipulate abstract data structures.

\appendix
\section{Appendix}
\label{sec:appendix}

In this short Appendix, we state and prove the ``vectorization lemma'' which we used in Section~\ref{sec:sorting}.

\begin{lem}
\label{lem:1}
Consider a matrix $\mat{M} \in \mathbb{R}^{n \times n}$, a vector $\vec{v} \in \mathbb{R}^n$, and the matrix-vector product $\trn{\mat{M}} \vec{v}$.

There exist a vector $\vec{m} \in \mathbb{R}^{n^2}$ and a matrix $\mat{V} \in \mathbb{R}^{n \times n^2}$ such that
\begin{equation*}
\trn{\mat{M}} \vec{v} = \mat{V} \vec{m}
\end{equation*}
In particular, this claim hods true for
\begin{align*}
\vec{m} & = \operatorname*{vec}\,(\mat{M}) \\
\mat{V} & = \mat{I} \otimes \trn{\vec{v}}
\end{align*}
where $\vec{m}$ contains the stacked columns of $\mat{M}$, $\mat{I}$ is the $n \times n$ identity matrix, and $\otimes$ denotes the Kronecker product.
\end{lem}

\begin{proof}
Let $\vec{r} = \trn{\mat{M}} \vec{v}$. Written in terms of the columns $\vec{m}_i$ of matrix $\mat{M} = [\vec{m}_1, \ldots, \vec{m}_n]$, the entries $r_i$ of vector $\vec{r}$ are given by $r_i = \ipt{\vec{m}_i}{\vec{v}}$ or, equivalently, $r_i  = \ipt{\vec{v}}{\vec{m}_i}$.

Introducing (much) larger vectors $\vec{m} \in \mathbb{R}^{n^2}$ and $\vec{u}_i \in \mathbb{R}^{n^2}$ where $\vec{m} = \operatorname*{vec}\,(\mat{M})$ and 
\begin{equation*}
\trn{\vec{u}_i} = \bigl[ \underbrace{\trn{\vec{0}} \, \cdots \, \trn{\vec{0}}}_{i-1 \text{ times}} \, \trn{\vec{v}} \, \underbrace{\trn{\vec{0}} \, \cdots \, \trn{\vec{0}}}_{n-i \text{ times}} \bigr] 
\end{equation*}
with $\vec{0} \in \mathbb{R}^n$ the vector of all $0$s, we can equivalently write $r_i = \ipt{\vec{u}_i}{\vec{m}}$.

Hence, when gathering the $i = 1, \ldots, n$ different vectors $\trn{\vec{u}_i}$ as the rows of a matrix $\mat{V} \in \mathbb{R}^{n \times n^2}$ such that
\begin{align*}
\mat{V} & = 
\begin{bmatrix}
\trn{\vec{v}} & \trn{\vec{0}} & \trn{\vec{0}} & \cdots & \trn{\vec{0}} & \trn{\vec{0}} \\
\trn{\vec{0}} & \trn{\vec{v}} & \trn{\vec{0}} & \cdots & \trn{\vec{0}} & \trn{\vec{0}} \\
\vdots & \vdots & \vdots & \ddots & \vdots & \vdots \\
\trn{\vec{0}} & \trn{\vec{0}} & \trn{\vec{0}} & \cdots & \trn{\vec{0}} & \trn{\vec{v}}
\end{bmatrix}
\end{align*}
we have $\mat{V} = \mat{I} \otimes \trn{\vec{v}}$ and observe that $\vec{r}$ can also be computed as $\vec{r} = \mat{V} \vec{m}$.
\end{proof}

\section*{Acknowledgments}

The authors gratefully acknowledge financial support by the Competence Center for Machine Learning Rhine-Ruhr (ML2R) which is funded by the Federal Ministry of Education and Research of Germany (grant no. 01IS18038C). 

\bibliographystyle{IEEEtran}
\bibliography{literature}

\end{document}